\documentclass{jgaa-art}
\usepackage{graphicx,xspace,amsmath,marvosym}
\usepackage{amsmath,amsfonts,amssymb,dsfont,bm,verbatim}
\usepackage{subfigure}

\graphicspath{{figures/}}

\newcommand{\R}{\ensuremath{\mathds R}}
\renewcommand{\P}{\ensuremath{\mathds P}}

\newcommand{\T}{\ensuremath{\mathcal T}}

\renewcommand{\phi}{\varphi}
\renewcommand{\epsilon}{\varepsilon}
\newcommand{\ER}{\exists\R}

\newtheorem{observation}{Observation}

\begin{document}

\HeadingAuthor{Cardinal and Kusters} 
\HeadingTitle{Complexity of Simultaneous Geometric Embedding} 
\title{The Complexity of Simultaneous Geometric Graph Embedding}%
\Ack{Vincent Kusters is partially supported by the ESF EUROCORES programme
  EuroGIGA, CRP GraDR and the Swiss National Science Foundation, SNF Project
  20GG21-134306. Jean Cardinal is partially supported by the ESF EUROCORES
  programme EuroGIGA, CRP ComPoSe.

  This paper previously appeared in: Journal of Graph Algorithms and
  Applications (JGAA), volume 19, number 1, 2015, pages 259--272 (doi:10.7155/jgaa.00356).}

\author[first]{Jean Cardinal}{jcardin@ulb.ac.be}
\author[second]{Vincent Kusters}{vincent.kusters@inf.ethz.ch}

\affiliation[first]{Computer Science Department, Universit\'e libre de Bruxelles (ULB), Belgium.}
\affiliation[second]{Department of Computer Science, ETH Z\"urich, Switzerland.}


\maketitle

\begin{abstract}
  Given a collection of planar graphs $G_1,\dots,G_k$ on the same
  set $V$ of $n$ vertices, the \emph{simultaneous geometric embedding (with mapping)} 
  problem, or simply $k$-SGE, is to find a set $P$ of $n$ points in the plane
  and a bijection $\phi: V \to P$ such that the induced straight-line drawings of
  $G_1,\dots,G_k$ under $\phi$ are all plane.

  This problem is polynomial-time equivalent to weak rectilinear realizability
  of abstract topological graphs, which
  Kyn\v{c}l~(doi:10.1007/s00454-010-9320-x) proved to be complete for $\ER$,
  the existential theory of the reals. Hence the problem $k$-SGE is
  polynomial-time equivalent to several other problems in computational
  geometry, such as recognizing intersection graphs of line segments or finding
  the rectilinear crossing number of a graph.

  We give an elementary reduction from the pseudoline stretchability problem to $k$-SGE, 
  with the property that both numbers $k$ and $n$ are linear in the number of pseudolines.
  This implies not only the $\ER$-hardness result, but also a $2^{2^{\Omega (n)}}$ lower bound 
  on the minimum size of a grid on which any such simultaneous embedding can be
  drawn. This bound is tight.
  Hence there exists such collections of graphs that can be simultaneously embedded, 
  but every simultaneous drawing requires an exponential number of bits per coordinates. 
  The best value that can be extracted from Kyn\v{c}l's proof is only $2^{2^{\Omega (\sqrt{n})}}$.
\end{abstract}

\Body 

\section{Introduction}
\label{sec:introduction}

Given graphs $G_1=(V,E_1),\dots,G_k=(V,E_k)$ on $n$ vertices,
\emph{simultaneous geometric embedding (with mapping)} or simply $k$-SGE, is
the problem of finding a point set $P\subset\R^2$ of size $n$ and a bijection
$\phi: V \to P$ such that the induced straight-line drawings of $G_1,\dots,G_k$
under $\phi$ are all plane~\cite{brass2007simultaneous}. The corresponding
decision problem (which we also refer to as $k$-SGE) simply asks whether such a
point set exists. It is important to note that $k$ is part of the input and can
thus depend on $n$. The problem 1-SGE amounts to planarity testing. The problem
$2$-SGE is typically referred to simply as SGE. Fig.~\ref{fig:sge_example}
shows an example of two graphs and a $2$-SGE.

\begin{figure}[h]
  \centering
  \includegraphics{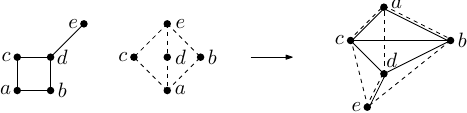}
  \caption{Graphs $G_1$ and $G_2$ on the same vertex set and a $2$-SGE of $G_1$ and $G_2$.}
  \label{fig:sge_example}
\end{figure}

Early work on the topic focused on the existence of $k$-SGEs for restricted
graph classes. The SGE problem was originally introduced by Brass
et.~al~\cite{brass2007simultaneous}. They show that there is a pair of
outerplanar graphs on the same vertex set that does not admit a $2$-SGE.
Additionally, they give a triple of paths that does not admit a $3$-SGE. The
authors also show that various other classes of graphs, such as a pair of
caterpillars, an extended star and a path, or two stars always admit a $2$-SGE.
The most recent positive result is a $2$-SGE construction that works for
generalizations of caterpillars with generalizations of stars, spiders and
caterpillars~\cite{digiacomo2014planar}. The question of whether any two trees
admit a $2$-SGE remained open for six years, until the question was settled in
the negative with a counterexample~\cite{geyer2009two}. The most recent
negative result gives a tree and a path that do not admit a
$2$-SGE~\cite{angelini2011tree}. Research has since focused on other
variations of the problem, such as \emph{simultaneous embedding with fixed
  edges} (edges are drawn as arbitrary simple curves, but all graphs
must use identical curves for identical edges), \emph{matched drawings}
(vertices have fixed $y$-coordinates in all drawings, but may have different
$x$-coordinates in each drawing), or \emph{partial simultaneous geometric
  embedding} (a limited number of vertices may be mapped to different points in
different drawings)~\cite{ColumnPlanarity_GD2014}. The decision problem $2$-SGE
is NP-hard~\cite{estrella2008simultaneous}.
See~\cite{blasius2013handbook} for an excellent survey.

The existential theory of the reals is the set of true sentences of the form
$\exists (x_1, \dots, x_n) : \phi(x_1, \dots, x_n)$, where $\phi$ is a
$(\wedge,\vee,\neg)$-formula over the signature $(0, 1, +, *, <,\leq,=)$
interpreted over the universe of real numbers~\cite{schaefer2013realizability}.
The decision problem ETR asks whether a given sentence is true. The complexity
class $\ER$ is defined as the set of decision problems that can be reduced to
ETR in polynomial time. A problem is $\ER$-hard if it is at least as hard as
every problem in $\ER$, i.e., if every problem in $\ER$ can be reduced to it in
polynomial time~\cite{schaefer2010complexity,basu2007existential}. A problem is
$\ER$-complete if it belongs to $\ER$ and is $\ER$-hard. It is known that
$\textrm{NP}\subseteq\ER\subseteq\textrm{PSPACE}$: Boolean satisfiability can
be encoded as a decision problem on a set of polynomial inequalities and
Canny~\cite{canny1988algebraic} gave a polynomial-space algorithm for ETR.

In 2011, Kyn\v{c}l~\cite{kyncl2011simple} proved that \emph{weak rectilinear
  realizability} of \emph{abstract topological graphs} is $\ER$-complete. Since
this problem reduces to $k$-SGE in polynomial
time~\cite{gassner2006simultaneous} and since $k$-SGE belongs to
$\ER$~\cite{estrella2008simultaneous}, it follows that $k$-SGE is
$\ER$-complete. The $k$-SGE problem is therefore polynomial-time equivalent to
many other classical problems in computational geometry, such as finding the
rectilinear crossing number of a graph~\cite{bienstock1991some}, recognizing
unit disk graphs~\cite{mcdiarmid2013integer}, recognizing intersection graphs
of convex sets in the plane~\cite{schaefer2010complexity}, recognizing
intersection graphs of segments~\cite{KM94,matousek2014intersection}, solving
the Steinitz problem~\cite{mnev1988universality}, and deciding the
realizability of linkages~\cite{kapovich2002universality}. We refer the reader
to recent work of Schaefer for more references and
examples~\cite{schaefer2010complexity,schaefer2013realizability}.

Our contribution is an elementary self-contained construction showing the $\ER$-hardness of $k$-SGE.
It involves a direct translation of the information contained in an arrangement of $n$
pseudolines into a set of $n$ planar graphs on a set $V$ of $O(n)$ vertices, in such a way that the
pseudoline arrangement is stretchable if and only if the graphs can be simultaneously embedded.
The main interesting feature of this construction is that the size of $V$ is linear 
in the number of pseudolines. This implies that for some positive instances of $k$-SGE, 
representing the point set by encoding the coordinates of each point requires an 
{\em exponential number of bits}. This follows from the analogous result on realizations 
of order types by Goodman, Pollack, and Sturmfels~\cite{goodman1989coordinate}. Our result improves on
Kyn\v{c}l's construction, which shows only that $2^{\Omega(\sqrt{n})}$ bits are
sometimes necessary.

In Section~\ref{sec:pseudolines_and_order_types}, we briefly recall standard
results on (realizability of) order types and (stretchability of) pseudoline
arrangements. The reduction itself is given in Section~\ref{sec:red}.
Section~\ref{sec:grid} presents our results on coordinate sizes in simultaneous
embeddings.

\section{Pseudolines and order types}
\label{sec:pseudolines_and_order_types}

Many combinatorial properties of a point set in the plane are captured by its
\emph{order type}. The order type of a point set $P\subset\R^2$ is the mapping
$\chi: {P\choose 3}\to \{-1,0,+1\}$, where
\[\chi (a, b, c) = \mathrm{sign}\left(\left|
\begin{array}{ccc}
a_x & a_y & 1\\
b_x & b_y & 1\\
c_x & c_y & 1\\
\end{array}
\right|\right).\]
The value of $\chi(a, b, c)$ determines whether the three points $a, b, c$ make
a left turn ($+1$), a right turn ($-1$), or are aligned ($0$). When
$\chi(a,b,c)\neq 0$ for all triples $a,b,c$, the point set is said to be in
\emph{general position} and $\chi$ is called \emph{uniform}. Among other
things, the order type encodes the convex hull of a point set and whether two
segments with endpoints in the point set intersect.

\emph{Abstract order types} generalize the notion of order types on planar
point sets. Knuth~\cite{knuth1992axioms} calls uniform abstract order types
\emph{CC-systems} and defines them by the following five axioms (we write
$\chi(p,q,r)$ instead of $\chi(p,q,r)=+1$ and $\neg\chi(p,q,r)$ instead of
$\chi(p,q,r)=-1$):
\begin{enumerate}
\item Cyclic symmetry: $\chi(p,q,r)\implies\chi(q,r,p)$.
\item Antisymmetry: $\chi(p,q,r)\implies\neg\chi(p,r,q)$.
\item Nondegeneracy: $\chi(p,q,r) \vee \chi(p,r,q)$.
\item Interiority: $\chi(t,q,r) \wedge \chi(p,t,r) \wedge \chi(p,q,t) \implies
  \chi(p,q,r)$.
\item Transitivity: $\chi(t,s,p)\wedge \chi(t,s,q)\wedge \chi(t,s,r)\wedge
  \chi(t,p,q)\wedge \chi(t,q,r)\implies \chi(t,p,r)$.
\end{enumerate}

Abstract order types are connected to the well-studied mathematical field of
oriented matroids. Specifically, if we consider the equivalence class where
$\chi=-\chi$, then Knuth~\cite{knuth1992axioms} proves that (equivalence
classes of) uniform abstract order types are in one-to-one correspondence with
\emph{uniform acyclic rank-3 oriented matroids}. We refer the interested reader
to~\cite{bjorner1999oriented} for more information on oriented matroids.
An abstract order type $\chi$ is \emph{realizable} if there exists a point set
in $\R^2$ with order type $\chi$. Not all abstract order types are realizable:
the smallest non-realizable abstract order type is the well-known Pappus
arrangement on 9 points.

Order types are closely related to \emph{pseudoline arrangements}.
Pseudoline arrangements are usually considered
in the real projective plane $\P^2$, where they can be defined as simple closed
curves, every pair of which meet in exactly one
point~\cite{grunbaum1972arrangements}. We recall that the projective plane is
the extension of the Euclidean plane by a point ``at infinity'' for each
direction $\alpha$ where the lines with direction $\alpha$ are defined to
intersect, and the line at infinity contains exactly the points at
infinity. For an excellent introduction to projective geometry, we refer the
interested reader to~\cite{richter2011perspectives}, but we do not assume any
familiarity with projective geometry here. Two projective pseudoline
arrangements $A$ and $A'$ in $\P^2$ are \emph{isomorphic} if there is a
self-homeomorphism of the projective plane that turns $A$ into $A'$.

(Uniform) abstract order types correspond exactly to (simple) projective
pseudoline arrangements with a marked face. For straight-line arrangements, the
marked face corresponds to the convex hull of the point set described by the
order type. For more background on pseudoline arrangements with a marked face
and their encodings, the reader is referred to Felsner~\cite{felsnerbook}
(Chapter 6).

By the Folkman-Lawrence topological representation
theorem~\cite{folkman1978oriented}, equivalence classes of projective
pseudoline arrangements correspond in one-to-one fashion to reorientation
classes of simple rank-3 oriented matroids~\cite{bjorner1999oriented}. A
pseudoline arrangement is \emph{simple} if no three pseudolines meet in the
same point. A simple projective pseudoline arrangement is
\emph{stretchable} if and only if it is isomorphic to a simple arrangement of
straight lines. In 1988, Mn\"ev proved that every semialgebraic set is stably
equivalent to the realization space of some rank-3 oriented
matroid~\cite{mnev1988universality}. Furthermore it was shown that the
underlying matroid could be made uniform (see also Lemma 4 in Shor~\cite{S91}).
As a by-product of these results, the \emph{simple pseudoline
stretchability} problem of deciding stretchability of a simple projective pseudoline
arrangement is $\ER$-complete~\cite{schaefer2010complexity}.

We define \emph{uniform order type realizability} as
the problem of deciding whether a given uniform abstract order type has a
realization. The following lemma summarizes the correspondence between abstract order types and pseudoline arrangements, and the polynomial-time equivalence of the realizability and stretchability problems. The order type realizability problem will be the starting point of our reduction.
\begin{lemma}
  \label{lem:realizable_iff_stretchable}
  Given a uniform abstract order type $\chi$, we can compute in polynomial time
  a description of a simple projective pseudoline arrangement $A$ with a marked face 
  such that $\chi$ is realizable if and only if $A$ is stretchable. Conversely, given a simple projective
  pseudoline arrangement $A$ with a marked face, we can compute in polynomial time a uniform
  abstract order type $\chi$ such that $\chi$ is realizable if and only if $A$
  is stretchable.
\end{lemma}

\section{$\ER$-completeness of $k$-SGE}
\label{sec:red}

We first reproduce the reduction from weak rectilinear realizability due to
Gassner et al.~\cite{gassner2006simultaneous} and then give a direct proof by
reduction from the stretchability problem.

\subsection{Reduction from weak rectilinear realizability}
\label{sub:reduction_from_weak_rectilinear_realizability}
An \emph{abstract topological graph} (AT-graph) is a pair $(G,R)$ where
$G=(V,E)$ is a graph and $R\subseteq{E\choose 2}$ is a set of pairs of its
edges. A straight-line drawing of $G$ is a \emph{weak rectilinear realization}
of $(G,R)$ if every pair of edges that cross in the drawing is contained in
$R$. Deciding if an AT-graph has a weak rectilinear realization was shown to be
$\ER$-complete by Kyn\v{c}l~\cite{kyncl2011simple}.

Kyn\v{c}l proves $\ER$-hardness of weak rectilinear realizability by a
reduction from simple pseudoline stretchability. Given a simple
arrangement $A$ of $m$ pseudolines, Kyn\v{c}l constructs an AT-graph $(G,R)$
with $G=(V,E)$ that admits a weak rectilinear realization if and only if $A$ is
stretchable. In this construction, similar to the \emph{order forcing lemma} of
Kratochv\'{\i}l and Matou\v{s}ek~\cite{KM94}, there is one edge associated with
each pseudoline, but there is also a pair of edges corresponding to each
crossing between two pseudolines.

%

The weak rectilinear realizability problem is closely related to the $k$-SGE
problem. The following equivalence is analogous to the equivalence given in
Theorem~2 of~\cite{gassner2006simultaneous}. Given graphs
$G_1=(V,E_1),\dots,G_k=(V,E_k)$, we construct an AT-graph $(G,R)$ with
$G=(V,\bigcup_i E_i)$ and $\{e,f\}\in R$ if and only if $\{e,f\}\not\subseteq
E_i$ for all $i$. Then $G_1,\dots,G_k$ admit a $k$-SGE if and only if $(G,R)$
admits a weak rectilinear realization. Conversely, given an AT-graph $(G,R)$
with $G=(V,E)$, we construct a graph $G_{ef}=(V,\{e,f\})$ for each pair of
edges $\{e,f\}\not\in R$. Then $(G,R)$ admits a weak rectilinear realization if
and only if the family $F=\{G_{ef} \mid \{e,f\}\not\in R\}$ admits an
$|F|$-SGE.

Combining Kyn\v{c}l's argument with this equivalence yields the following:
given an arrangement $A$ of $m$ pseudolines, we can construct a set of $k_m$
graphs $G_i$, each on $n_m$ vertices, such that $G_1,\dots,G_{k_m}$ admit a
$k_m$-SGE if and only if $A$ is stretchable. Here, $k_m=\Theta(m^4)$ and
$n_m=\Theta(m^2)$. Fix any constant $0<\epsilon\leq 1$ and add an additional
$m^{4/\epsilon}$ isolated vertices to each graph $G_i$. After this modification,
$n_m=\Theta(m^{4/\epsilon})$ and thus $k_m=\Theta(n_m^\epsilon)$. Since this
takes polynomial time, we obtain the following:
\begin{theorem}
  \label{thm:sge_er_complete_weak}
  Given graphs $G_1=(V,E_1),\dots,G_k=(V,E_k)$ on $n$ vertices, the decision
  problem $k$-SGE is $\ER$-complete for $k=\Omega(n^\epsilon)$ and any constant
  $\epsilon>0$.
\end{theorem}

\subsection{Reduction from uniform order type realizability}
We will give an alternative proof of the result from the previous section via a
polynomial-time reduction from uniform order type realizability to $k$-SGE.

\begin{figure}[b]
  \centering
  \subfigure[\label{fig:radial_def_ccw}]{\includegraphics{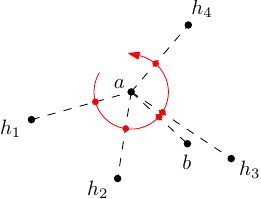}}\quad
  \subfigure[\label{fig:dual_radial}]{\includegraphics{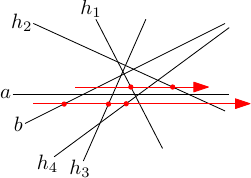}}
  \caption{(a) The radial ordering around point $a$ is $h_1, h_2, b, h_3, h_4$
    (up to a circular shift)~\cite{ReconstructingRadial_ISAAC2014}. (b) In the dual, $h_1,h_2$
    intersect $a$ from above in this order and $b,h_3,h_4$ intersect $a$ from
    below in this order.}
  \label{fig:radial_primal_dual}
\end{figure}

\subsubsection{Radial Systems}
The high-level idea is to associate one graph with each point $v$ of the
uniform abstract order type $\chi$ so that a proper geometric embedding of the
graph forces some \emph{radial ordering} of the other points around $v$. In a
set of $n$ points in the plane, the radial ordering $R(v)$ associated with the
point $v$ is simply the order in which the $n-1$ other points are encountered
by a counterclockwise ray sweep around $v$. We will define this order up to a
circular shift. This is illustrated in Fig.~\ref{fig:radial_def_ccw}. We call
$R$ the \emph{(counterclockwise) radial system} of $\chi$.

The radial system can be inferred from the order type of the point set, and
therefore can be defined even if $\chi$ is not realizable. A way to determine
the radial ordering of $v$ is to pick another point $w$ and first consider only
the points $x$ such that $\chi (v, w, x)=+1$, that is, the points on the left
of the oriented line $vw$. We can then sort these points using that $x<x'$
whenever $\chi (x, v, x') =-1$. We can similarly sort the points $x$ such that
$\chi(v, w, x)=+1$ and recover the complete radial ordering $R(v)$.

The radial system can also be extracted from the pseudoline arrangement
corresponding to the abstract order type. For the pseudoline $v$, the
corresponding radial ordering is constructed as follows. Starting on any point
on the pseudoline $v$, we first report the order of the successive
intersections with pseudolines coming from the same side as the marked face, followed
by the order of the intersections with the pseudolines coming from the other side.
In a Euclidean realization of the arrangement, this corresponds to pseudolines
coming from {\em above} and from {\em below}, respectively. 
This dual definition of the radial
orderings is illustrated on Fig.~\ref{fig:dual_radial}. Note that the radial
orderings are not the same as the {\em local sequences} defined by Goodman and
Pollack~\cite{goodman1984semispaces}. The local sequence for a pseudoline $v$
is simply the order of the intersections with the other pseudolines, and
correspond in the primal point set to a sweep with a {\em line} through the
point $v$, instead of a ray.

The relation between radial systems and order types has been studied in depth
in a more general setting in a recent paper from a superset of the current
authors~\cite{ReconstructingRadial_ISAAC2014}. It was shown in particular that the radial orderings
alone are {\em not sufficient} to recover the complete order type of a point
set in the plane. Furthermore, for point sets with a triangular convex hull,
there can be as many as $n-1$ different order types having the exact same
radial orderings for each point. The reader is referred to this paper for more
results and examples.

It is not too difficult to show, however, that the set of radial orderings is
sufficient to recover the order type, provided we also know the points on the
convex hull. This is a specialization of Lemma~3 in~\cite{ReconstructingRadial_ISAAC2014}.
\begin{lemma}[\cite{ReconstructingRadial_ISAAC2014}]
  \label{lem:ch+ro}
  Consider a realizable abstract order type $\chi$ on $n$ points, let $S$ be
  the set of counterclockwise radial orderings of the points, and let $H$ be
  the set of points on the convex hull. Then the pair $(S, H)$ uniquely
  determines $\chi$.
\end{lemma}
The proof is straightforward and involves three steps. We first recover the
order of the points on the convex hull by looking at the radial ordering of one
of them. Next, we recover the orientation of every triple with at least one
point $p$ on the convex hull from the radial ordering of $p$. Finally, the
orientations of the remaining triples are deduced by sweeping a ray around a
point on the convex hull and then sweeping a ray around every point
encountered.

\subsubsection{Reduction}

Before delving into the construction, we have to argue that we can assume
without loss of generality that the convex hull of the input abstract order
type $\chi$ for the realizability problem is triangular. 
Using Lemma~\ref{lem:realizable_iff_stretchable}, we can compute a projective
pseudoline arrangement $A$ that is stretchable if and only if $\chi$ is
realizable, in such a way that the convex hull of $\chi$ corresponds to the
marked face of $A$. If this face is bounded by at
most three pseudolines, then we are done. Otherwise, since it is known that every
projective arrangement of $n$ pseudolines has at least $n$ triangular
faces~\cite{levi1926teilung}, we make such a face the marked face of $A$.
Applying Lemma~\ref{lem:realizable_iff_stretchable} in the other direction finally gives
us an abstract order type $\chi'$ with a triangular convex hull that is
realizable if and only if $\chi$ is realizable.

We now have all the ingredients required for the reduction. For each $v\in V$ we
define the \emph{wheel graph} $W_v$ on $V$ as the union of the cycle $R(v)$ 
corresponding to the radial ordering around $v$, and the star connecting $v$ 
to all vertices in $R(v)$. The purpose of including such a graph is to encode the
radial ordering $R(v)$ of the $n-1$ other points around $v$.

We next create the labeled graph $T_v$ by embedding three copies of $W_v$ into 
the interior faces of a copy of $K_4$, the 
complete graph on four vertices $\{t_1,t_2,t_3,t_4\}$, as shown on the left 
in Fig.~\ref{fig:wheel_frame}. 
We distinguish the vertices of different copies by adding a superscript $i$ to the 
vertices of copy $i$. The
convex hull $h_1^1,h_2^1,h_3^1$ is embedded onto $t_1,t_4,t_3$; the convex hull
$h_1^2,h_2^2,h_3^2$ is embedded onto $t_2,t_4,t_1$; and the convex hull
$h_1^3,h_2^3,h_3^3$ is embedded onto $t_3,t_4,t_2$. Fig.~\ref{fig:wheel_frame}
shows an example of a wheel graph $W_v$ and the resulting graph $T_v$. The graph
$T_v$ has exactly $3n-5$ vertices. The reason why we need to embed three copies of
the wheel graph $W_v$, and not simply one, is that the abstract order type will be
preserved only provided the convex hull is the same. We will see that three copies
are sufficient to guarantee that at least one of them will have the same convex hull
as the one specified by the original abstract order type. 

\begin{figure}[t]
  \centering
  \includegraphics{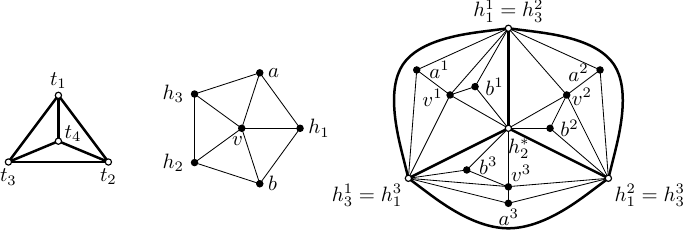}
  \caption{From left to right: the copy of $K_4$, a wheel graph $W_v$ for an
    order type with convex hull $h_1,h_2,h_3$, and the graph $T_v$. In $T_v$ we
    have $t_1=h_1^1=h_3^2$, $t_2=h_1^2=h_3^3$, $t_3=h_3^1=h_1^3$, and
    $t_4=h_2^1=h_2^2=h_2^3$.}
  \label{fig:wheel_frame}
\end{figure}

Though the $T_v$ in the example is maximal planar, this is not always the case.
We do, however, have the following.
\begin{lemma}
  \label{lem:tv_3_connected}
  Each $T_v$ is $3$-connected.
\end{lemma}
\begin{proof}
  Using symmetry it is easy to verify that $W_v$ is $3$-connected.
  We will use Menger's theorem to prove that $T_v$ is also $3$-connected. Let
  $u$ be any vertex of $W_v$. From every vertex $u$ in $W_v$ there is a path to
  $h_1$, a path to $h_2$ and a path to $h_3$ such that the paths share no vertex
  other than $u$. This can be seen as follows. If $u=v$ then we can reach each
  $h_i$ in one step. Otherwise, one path traverses the cycle in a clockwise
  direction, one traverses the cycle in a counterclockwise direction and one
  goes via $v$. The same holds for the copy of $K_4$ (which can also be thought 
  of as a wheel graph). It follows immediately that there are three interior pairwise
  vertex-disjoint paths between every two vertices in $T_v$. Hence, the lemma
  follows by Menger's theorem.
\end{proof}
Since $T_v$ is $3$-connected, all embeddings of $T_v$ are the same up to
reflection and the choice of the outer face. Let $\T$ be the set of all $T_v$.
\begin{theorem}
  \label{thm:sge_iff_realizable}
  Given a abstract order type $\chi$ on a set $V$ of $n$ elements, we can compute in
  polynomial time a set $\T$ of $n$ graphs, each on the same set of $3n-5$
  vertices, such that $\T$ admits an $n$-SGE if and only if $\chi$ is
  realizable.
\end{theorem}
\begin{proof}
  Suppose that $\chi$ is realizable. Let $P$ be a labeled point set that
  realizes $\chi$ and let $p(v)$ be the point in $P$ that corresponds to $v$ in
  $\chi$. After possibly reflecting $P$ along the $y$-axis, the
  counterclockwise ray sweep around each point $p(v)\in P$ encounters the other
  points of $P$ in the order $R(v)$. Hence, by construction of the wheel
  graphs, the induced straight-line drawing of each $W_v$ on $P$ is plane. A
  labeled point set whose induced straight-line drawing of each $T_v$ is plane
  can now easily be constructed from three copies of $P$ and affine
  transformations.

  Conversely, suppose that $\T$ has an $n$-SGE $\phi$ and consider its
  convex hull. Note that the convex hull of $\phi$ corresponds to a mutual face
  of all $T_v$: if some $T_v$ does not have a face that corresponds to the
  convex hull, then some vertex of $T_v$ must have been embedded in the outer
  face of $T_v$ in $\phi$, which is impossible by
  Lemma~\ref{lem:tv_3_connected}. If $t_1,t_2,t_3$ is the (clockwise) outer face
  in $\phi$, then the point set corresponding to one of the three copies is a
  realization of $\chi$. This can be seen as follows. If $t_1,t_2,t_3$ is the
  outer face in this clockwise order, then the triangle $h_1^1,h_2^1,h_3^1$ is
  also oriented in this clockwise order in $\phi$. By
  Lemma~\ref{lem:tv_3_connected}, this triangle must form the convex hull of
  each $W_v^1$. Hence, any swap of two elements in any radial ordering
  $R(v^1)$ in $\phi$ will induce a crossing in the drawing of $W_v^1$. It
  follows that $\phi$ is consistent with all radial orderings and therefore, 
  from Lemma~\ref{lem:ch+ro}, it induces a realization of $\chi$. 
  If a face other than $t_1,t_2,t_3$ was chosen to be the
  outer face in $\phi$, say a face bounded by three vertices of copy one, then
  the point set corresponding to the vertices of the second copy (or the third; both work)
  is a realization of $\chi$ by a similar argument. This concludes the proof.
\end{proof}
We showed that uniform order type realizability can be reduced in polynomial
time to $k$-SGE. Since uniform order type realizability is $\ER$-complete, it
follows that $k$-SGE is $\ER$-hard and hence $\ER$-complete by the fact that
$k$-SGE belongs to $\ER$~\cite{estrella2008simultaneous}. We finally add a
suitable amount of isolated vertices as explained in
Subsection~\ref{sub:reduction_from_weak_rectilinear_realizability} to complete
our alternative proof of Theorem~\ref{thm:sge_er_complete_weak}.

We define \emph{radial system realizability} as the problem of deciding whether
a given system of permutations $R$ is the radial system of a set of points in
$\R^2$.

\begin{observation}
  \label{obs:radial_system_realizability}
  Radial system realizability is $\ER$-complete.
\end{observation}
\begin{proof}
  We prove that radial system realizability is polynomially equivalent to
  uniform order type realizability.

  Consider an abstract order type $\chi$. Using the method described at the
  beginning of this section, compute an abstract order type $\chi'$ from $\chi$
  in polynomial time such that $\chi'$ is realizable if and only if $\chi$ is
  realizable and $\chi'$ has a triangular convex hull. Compute the radial
  system $R$ of $\chi'$. Since $\chi'$ has a triangular convex hull, $\chi'$ is
  the only abstract order type with radial system
  $R$~(Theorem~1~in~\cite{ReconstructingRadial_ISAAC2014}). Hence, $R$ is
  realizable if and only if $\chi$ is realizable.

  Conversely, consider a system of permutations $R$ on a set of $n$ elements.
  We compute in polynomial time the set $T(R)$ of at most $n-1$ uniform
  abstract order types that have $R$ as their radial system~(Theorem~1 and
  Corollary~1 in~\cite{ReconstructingRadial_ISAAC2014}). Then $R$ is realizable
  if and only if at least one uniform abstract order type in $T(R)$ is
  realizable.
\end{proof}

\section{Simultaneous geometric embeddings requiring doubly exponential grids}
\label{sec:grid}

Many graph drawing questions involve drawing graphs on a small grid. Our
construction gives insight on the following simple problem: given a collection
of $k$ graphs on $n$ vertices which admit a $k$-SGE, can we provide any
guarantee on the size of the largest grid on which we can embed them?

Since the decision problem is $\ER$-hard, it is unlikely that simultaneous
geometric embeddings can all be drawn on a small grid. In fact, showing that
all collections of $k=\mathrm{poly}(n)$ graphs that admit a $k$-SGE, admit a
$k$-SGE on a grid of size at most exponential in a polynomial in $n$ would
directly imply that $\ER = NP$, since the drawing could be encoded using a
polynomial number of bits and used as a certificate.

Our construction implies the following lower bound on the size of the smallest
grid required for a $k$-SGE.
\begin{theorem}
  \label{th:large_coords}
  There exist collections of $k$ graphs on $n=\Theta(k)$ vertices that admit a
  $k$-SGE, every $k$-SGE of which requires a grid of size $2^{2^{\Omega (n)}}$.
\end{theorem}
\begin{proof}
  We use a well-known construction due to Goodman, Pollack, and
  Sturmfels~\cite{goodman1989coordinate}. They construct a
  set of $k$ points in the plane such that every realization of its order type
  requires a grid of size $2^{2^{\Omega (k)}}$. This construction implements an
  iterative squaring procedure using the multiplication gadget from Von
  Staudt's algebra of throws~\cite{RG95,mcdiarmid2013integer}.

  Let $\chi$ be an order type on such a set of $k$ points requiring a doubly
  exponential-size grid. By Theorem~\ref{thm:sge_iff_realizable}, we can
  construct graphs $T_1,\dots,T_k$ where each $T_i$ has $n=3k-5$ vertices, such
  that every point set $P$ that admits an $k$-SGE of $T_1,\dots,T_k$, where
  $|P|=n$, will contain a copy of a realization of $\chi$. By definition of
  $\chi$, $P$ cannot be represented with points of integer coordinates smaller
  than $2^{2^{\Omega (n)}}$.
\end{proof}
In the same paper, Goodman, Pollack, and Sturmfels~\cite{goodman1989coordinate} prove that
every realizable order type in general position has a realization with
coordinates bounded by $2^{2^{O(n)}}$. If a set of graphs admits a $k$-SGE,
then the resulting point set can perturbed into general position without
introducing any crossing. Finally, since the order type of a point set
determines whether two segments cross, it follows that a $k$-SGE never requires
coordinates larger than $2^{2^{O(n)}}$. Hence, Theorem~\ref{th:large_coords} is
tight.

This is a significant improvement compared to what can be extracted from the
construction of Kyn\v{c}l~\cite{kyncl2011simple}. In the latter, an arrangement
of $m$ pseudolines is realized via the simultaneous geometric embedding of
graphs on a set of $k_m$ graphs on $n_m=\Theta(m^2)$ vertices. Hence although
the coordinates may have value $2^{2^{\Omega(m)}}$, this is only
$2^{2^{\Omega(\sqrt{n_m})}}$.

\section{Concluding remarks}
\label{sec:concluding_remarks}

We gave an alternative proof for the $\ER$-hardness of $k$-SGE and we showed
that a $k$-SGE may sometimes need a grid of size $2^{2^{\Omega(n)}}$. Our
hardness proof relies on choosing $k=\Omega(n^\epsilon)$, and it is not clear
how to weaken this requirement. The complexity of the cases $k=O(\log n)$ and
in particular $k=2$ are still open, including whether these problems are in NP.

\paragraph{Acknowledgments.}
This work was initiated during a visit of the second author in Brussels,
supported by the EUROCORES programme EUROGIGA, CRP ComPoSe. The authors
gratefully acknowledge discussions with Stefan Langerman, who gave insightful
comments on preliminary versions of the results. The authors would like to
thank Marcus Schaefer for pointing out the relevance of~\cite{kyncl2011simple}
and the anonymous reviewers for their helpful suggestions.

\bibliographystyle{abbrvurl}
\bibliography{bibliography}
\end{document}